\documentclass[letterpaper, 10 pt, conference]{ieeeconf}
\IEEEoverridecommandlockouts 
\usepackage{cite}
\usepackage{amsmath,amssymb,amsfonts}
\usepackage{graphicx}
\usepackage{textcomp}
\usepackage{amsfonts, amsmath}
\usepackage{amssymb}
\usepackage{cleveref}
\usepackage{amsmath}
\usepackage{graphicx}

\newtheorem{theorem}{Theorem}
\newtheorem{lemma}{Lemma}

\newtheorem{assumption}{Assumption}
\usepackage{cuted}
\usepackage[english]{babel}
\usepackage{xcolor}
    
    \makeatletter
 \let\old@ps@headings\ps@headings
 \let\old@ps@IEEEtitlepagestyle\ps@IEEEtitlepagestyle
 \def\confheader#1{%
 \def\ps@headings{%
 \old@ps@headings%
 \def\@oddhead{\strut\hfill#1\hfill\strut}%
 \def\@evenhead{\strut\hfill#1\hfill\strut}%
 }%
 \def\ps@IEEEtitlepagestyle{%
 \old@ps@IEEEtitlepagestyle%
 \def\@oddhead{\strut\hfill#1\hfill\strut}%
 \def\@evenhead{\strut\hfill#1\hfill\strut}%
 }%
 \ps@headings%
 }
 \makeatother

 \usepackage[pscoord]{eso-pic}

  \makeatletter
\newcommand{\linebreakand}{%
  \end{@IEEEauthorhalign}
  \hfill\mbox{}\par
  \mbox{}\hfill\begin{@IEEEauthorhalign}
}
\makeatother  

\begin{document}

\title{State and Input Constrained Model Reference Adaptive Control 
}

\author{Poulomee~Ghosh and Shubhendu~Bhasin
\thanks{Poulomee Ghosh and Shubhendu Bhasin are with Department of Electrical Engineering, Indian Institute of Technology Delhi, New Delhi, India. 
        {\tt\small (Email: Poulomee.Ghosh@ee.iitd.ac.in, sbhasin@ee.iitd.ac.in)}}}




\maketitle

\begin{abstract}
Satisfaction of state and input constraints is one of the most critical requirements in control engineering applications. In classical model reference adaptive control (MRAC) formulation, although the states and the input remain bounded, the bound is neither user-defined nor known a-priori. In this paper, an MRAC is developed for multivariable linear time-invariant (LTI) plant with user-defined state and input constraints using a simple saturated control design coupled with a barrier Lyapunov function (BLF). Without any restrictive assumptions that may limit practical implementation, the proposed controller guarantees that both the plant state and the control input remain within a user-defined safe set for all time while simultaneously ensuring that the plant state trajectory tracks the reference model trajectory. The controller ensures that all the closed-loop signals remain bounded and the trajectory tracking error converges to zero asymptotically.
Simulation results validate the efficacy of the proposed constrained MRAC in terms of better tracking performance and limited control effort compared to the standard MRAC algorithm.

\end{abstract}


\section{Introduction}
\label{sec:intro}
Modern control applications are characterized by physical, safety and energy limitations that can often be translated into state and input constraints on the plant dynamics. Without accounting for their effect during control design, these constraints are often met using an ad-hoc approach. While  conventional adaptive control techniques manage to control systems under parametric uncertainty, they fail to adhere to user-defined constraints. This challenge is exacerbated for safety critical systems where safety is of utmost importance. \\
A particular versatile class of adaptive controllers is model reference adaptive controllers (MRAC) that aim to control systems with parametric and matched uncertainties by tracking a user-defined stable reference model system \cite{classMRAC, mrac2, mracnew,  mrac3}. Although the control law ensures asymptotic tracking, the bound on the tracking error is typically not known a-priori. Moreover, classical MRAC does not allow for user-defined constraints on the state and the input. 
Large magnitude of control effort might exceed actuator's saturation limit and cause deterioration of the process. Hence, constraining the plant states and input within known user-defined bounds while meeting satisfactory performance objectives is a problem of practical interest.\\
Several control approaches have been proposed that either partially or fully address these challenges. The state constrained tracking control problem has been dealt with using model predictive control \cite{mpc1},\cite{mpc}, optimal control theory \cite{opt1},\cite{opt}, invariant set theory \cite{blanchini},\cite{set}, reference governor approach \cite{rga}, \cite{gilbert} etc. Most of these techniques involve solving an optimization problem at every time instant which can be computationally expensive. To deal with these problems, various tools like Barrier Lyapunov Function (BLF) \cite{BLF}, Control Barrier Function(CBF) \cite{CBF} etc. have been utilized to guarantee safety. 
The safety certificate of CBF was introduced in \cite{CBF} where CBF and control Lyapunov function (CLF) are unified through quadratic programs (CBF-CLF-QP approach) to ensure safety in terms of forward invariance of set \cite{CBF2}, \cite{CBF3}. For CBF-CLF-QP approach, although the controller guarantees safety, the system trajectory doesn't essentially converge to the origin. A generalized approach of BLF was presented in \cite{BLF},\cite{BLF2} to satisfy safety
constraints for output feedback control systems. In \cite{Lafflitto}, BLF is used with model reference
adaptive controller for constraining trajectory tracking error and adaptive gains within user-defined sets where the estimate of the controller parameters is assumed to be very close to the actual value. An alternative approach to ensure safety is the state transformation technique using BLF. An adaptive tracking controller is developed in transformed state space using BLF to guarantee performance bound in \cite{stc}. Although BLF-based controllers ensure that user-defined state constraints are met, they usually result in large control effort when the states approach the boundary of the constrained region, often violating the actuator's operating limits.\\
Constraining the input to account for actuator saturation limits is another issue of practical concern that has been tackled extensively in literature, especially using various saturated functions e.g. hyperbolic tangent, sigmoid etc. \cite{incon10, inconnew, incon11, incon12, incon13}. 
An adaptive controller is developed
in \cite{annaswamy}, \cite{Lav} for a single-input single-output (SISO) LTI plant in the presence of input constraints. In \cite{incon1}, an adaptive tracking control method has been investigated for MIMO nonlinear systems where an auxiliary design system was introduced to deal with input constraints. \\
All the aforementioned approaches either involve state or input constraints. Few control approaches exist that constrain both state and input for uncertain systems. MPC \cite{dhar, mpc11, dhar2} is a popular control approach where both state and input constraints can be included in the optimization routine, albeit at the cost of computational complexity. Further, limited or imperfect model knowledge often leads to conservative results. A recent work \cite{Anderson} develops an MRAC law that places user-defined bounds on state and input. The result, however, is achieved by developing an auxiliary reference model that complicates the analysis and design of adaptive laws.\\
The main contribution of the work is the development of MRAC for multivariable LTI systems while satisfying safety constraints on both the state and the input. To this end, a saturated controller with BLF-based adaptation laws are designed. Inspired by \cite{annaswamy}, the design of the saturated controller is intuitive and simple, while the classical MRAC adaptation laws are modified based on BLF to ensure that the states and the input always lie in a user-defined safe region. Closed-loop signals are guaranteed to be bounded and the trajectory tracking error can be proved to converge to zero asymptotically. The adaptation rate being unchanged, the proposed controller shows better tracking performance than the classical MRAC while simultaneously guaranteeing the safety constraints.\\
This paper is organized as follows.  Section II presents the problem formulation and preliminaries of standard MRAC framework, section III elucidates proposed BLF-based methodology which satisfies state and input constraints. To justify the preeminence of the proposed controller, simulation results and comparative analysis are shown in section IV while section V comprises of conclusion and future works.



\section{Problem Formulation}
\label{sec:LFSR}
Throughout this paper $\mathbb{R}$ denotes the set of real numbers, $\mathbb{R}^{p \times q}$ denotes set of $p\times q$ real matrices, the identity matrix in $\mathbb{R}^{p \times p}$ is denoted by $I_{p}$ and $\|.\|$ represents the Euclidian vector norm and corresponding equi-induced matrix norm.

\subsection{Problem Statement}

Consider a linear time-invariant system
\begin{align}
    \dot{x}(t)=Ax(t)+Bu(t)
    \label{plant}
\end{align}
where $x(t)\in \mathbb{R}^n$ denotes the system state, $u(t) \in \mathbb{R}^m$ denotes the control input, $A \in \mathbb{R}^{n \times n}$ is the unknown system matrix, and  $B \in \mathbb{R}^{n \times m}$ is the input matrix, assumed to be full rank and known. The pair (A,B) is assumed to be stabilizable.\\
A reference model is considered as
\begin{align}
    \dot{x}_r(t)=A_rx_r(t)+B_rr(t)
    \label{ref}
\end{align}
where $x_r(t)\in \mathbb{R}^n$ is the state of the reference model, $r(t) \in \mathbb{R}^{m}$ is a bounded piecewise continuous reference input, $A_r \in \mathbb{R}^{n \times n}$, $B_r \in \mathbb{R}^{n \times m}$ are known. It is assumed that $A_r$ is Hurwitz i.e. for every $Q=Q^T>0$, there exists $P=P^T>0$ such that $A_r^TP+PA_r+Q=0$. 
\\~\\
\textbf{State Constraint:}
For any positive constant $\beta$, the system states should remain within a user defined safe set given by $\Omega_x:=\{x\in \mathbb{R}^n:\|x\|\leq \beta\}.$
\\~\\
\textbf{Input Constraint:} Magnitude of the control input should remain bounded in a safe set given by $\Omega_u := \{u\in\mathbb{R}^{m}: \|u\|\leq u_{max}\}$, where $u_{max}$ is a user-defined positive constant.
\begin{assumption}
For any $\beta>0$, there exist positive constants $\alpha_1,\alpha_2\in \mathbb{R}$ such that
\begin{align}
    & \|x_r(t)\| \leq \alpha_1 < \beta\\
    & \|\dot{x}_r(t)\|\leq \alpha_2
    \end{align}
\end{assumption}

\begin{assumption}
There exists a feasible control policy $u(t)$ that satisfies both the input and state constraints for all time. 
\end{assumption}

\textbf{Control Objective:}
 The control objective is to design an input $u(t)$, such that $x(t)$ tracks $x_r(t)$ i.e. $e(t) \triangleq x(t)-x_r(t)\rightarrow 0$ as $t \rightarrow \infty$ while both the input and the state constraints are satisfied. Using Assumption 1, the state constraints can be transformed to the constraint on the tracking error: $\|e(t)\|<k_b$, $\forall t\geq 0$, where $k_b \in \mathbb{R}$ is a positive constant given by $k_b=\beta-\alpha_1$, i.e. $\|e(t)\|<k_b \implies\|x(t)\|\leq\beta$.

\subsection{Classical MRAC}

Consider the classical certainty equivalence adaptive control law 
\begin{align}
    u=\hat{K}_xx+\hat{K}_rr
    \label{ueq}
\end{align}
where $\hat{K}_x(t)\in \mathbb{R}^{m\times n}$ and  $\hat{K}_r(t)\in \mathbb{R}^{m\times m}$ are estimates of the controller parameters $K_x$ and $K_r$, respectively which satisfy the matching conditions given in Assumption 2.
\begin{assumption}
There exists controller parameters $K_x\in \mathbb{R}^{m\times n}$ and $K_r\in \mathbb{R}^{m\times m}$ such that the following matching conditions are satisfied.
\begin{align}
    &A+BK_x=A_r \nonumber\\
    &BK_r=B_r
    \label{mc}
\end{align}
\end{assumption}
Using (\ref{ueq}) and (\ref{mc}) , the closed-loop error dynamics can be obtained as
\begin{align}
  \dot{e}=A_re+B\tilde{K}_xx+B\tilde{K}_rr
\end{align}
where $\tilde{K}_x(t)\triangleq \hat{K_x}(t)- K_x\in\mathbb{R}^{m \times n}$ and $\tilde{K}_r(t)\triangleq \hat{K_r}(t) - K_r\in \mathbb{R}^{m \times m}$ denote the parameter estimation errors. 

%
The classical adaptive update laws are given as \cite{classMRAC}
\begin{align}
  &\dot{\hat{K}}_x=-{\Gamma_xB^TPex^T}\nonumber\\
    &\dot{\hat{K}}_r=-{\Gamma_r B^TPer^T} 
    \label{MRAC}
\end{align}
where, $\Gamma_x\in \mathbb{R}^{m \times m}$ and $\Gamma_r \in \mathbb{R}^{m \times m}$ are  positive definite adaptation gain matrices. using Lyapunov analysis, we can prove that all the closed-loop signals remain bounded and the trajectory tracking error converges to zero asymptotically \cite{classMRAC},\cite{mrac2}.\\
Note that, actuator limits and state constraints are typically not considered in the classical MRAC design, rather they are often imposed in an ad-hoc manner during implementation. The focus of the work is to consider both input and state constraints in the MRAC design procedure and rigorously show boundedness of closed-loop signals and provide stability guarantees.

\section{Proposed Methodology}

\subsection{Input Constraint Satisfaction Using Saturated Control Design}

Consider the linear time-invariant plant given in (\ref{plant}) and reference model given in (\ref{ref}). 
Consider an auxiliary control input $v(t)\in \mathbb{R}^m$ as
\begin{align}
&v(t)=\hat{K}_xx+\hat{K}_rr
\label{pc1}
\end{align}
where $v(t)\triangleq[v_1(t),\hdots, v_m(t)]^T $. Inspired by \cite{annaswamy}, the saturated feedback controller is designed for the vector case as
\begin{align}
&u_i(t)= \begin{cases}
v_i(t) & \text{if}\:\:\: |v_i(t)|\leq \frac{u_{max}}{\sqrt{m}}\\
\frac{u_{max}}{\sqrt{m}}sgn(v_i(t)) & \text{if}\:\:\: |v_i(t)|>\frac{u_{max}}{\sqrt{m}}
\end{cases},
&& i=1, \hdots , m
\label{pc2}
\end{align}
where $u(t)\triangleq[u_1(t),\hdots, u_m(t)]^T $.  
The closed-loop error dynamics is given as
\begin{align}
    \dot{e}=A_re+B\tilde{K}_xx+B\tilde{K}_rr+B\Delta u
    \label{edot}
\end{align}
where  $\Delta u(t)\in\mathbb{R}^{m}$ is defined as the difference between actual control input and auxiliary control input : $\Delta u(t) \triangleq u(t)-v(t)$. To mitigate the effect of the disturbance term $B\Delta u(t)$, consider an auxiliary error signal $e_1(t)$ with the following dynamics
\begin{align}
    \dot{e}_1=A_re_1+K_1(t)\Delta u && e_1(t_0)=0
\end{align}
where $K_1(t)\in \mathbb{R}^{n \times m}$ is time-varying controller parameter. Let $e_d(t)$ be the difference between the actual and auxiliary error signals: $e_d(t) \triangleq e(t)-e_1(t)$. The error dynamics of $e_d(t)$ is given by
\begin{align}
  \dot{e}_d=A_re_d+B\tilde{K}_xx+B\tilde{K}_rr+K_d\Delta u
\end{align}
where $K_d(t)\triangleq B-K_1(t)$.
\\
\subsection{State Constraint Satisfaction using BLF}\label{AA}

To ensure state constraint satisfaction, a BLF is introduced \cite{BLF}.

\begin{assumption}
The initial conditions $x(0)$ are chosen such that the initial trajectory tracking error is bounded as $\|e(0)\|<k_b$.
\end{assumption}




\begin{lemma}
For any positive constant $k_b$, let $\Omega_e := \{e \in \mathbb{R}^n : \|e\|<k_b \}\subset \mathbb{R}^n$ and $\Psi:=\mathbb{R}^N\times \zeta \subset \mathbb{R}^{N+n}$ be open sets. Consider the system dynamics given by
\begin{align}
    \dot{\mu}=f(t,\mu)
\end{align}
$\mu:=[e^T,\xi^T]^T\in \Psi$, where $\xi:=[\tilde{K}_x,\tilde{K}_r,K_d,K_1]$ is the augmentation of the unconstrained states and the function $f:\mathbb{R}_{+}\times \Psi \rightarrow \mathbb{R}^{N+n}$ is measurable  for each fixed $\mu$ and locally Lipschitz in $e$, piecewise continuous and locally integrable on $t$. Suppose, there exists positive definite, decrescent, quadratic candidate Lyapunov function $V_2(\xi):\mathbb{R}^N\rightarrow\mathbb{R}$ and continuously differentiable, positive definite, scalar function $V_1(e):\Omega_e \rightarrow \mathbb{R}$, defined in an open region containing the origin such that 
\begin{align}
    V_1(e)\rightarrow \infty && \|e\|\rightarrow k_b
\end{align}
The candidate Lyapunov function can be written as $V(\mu)=V_1(e)+V_2(\xi)$. Given Assumption 3, if the following inequality holds
\begin{align}
    \dot{V}=\frac{\partial V}{\partial \mu}f\leq 0
\end{align}
then $e(t)\in \Omega_e \forall t$.
\begin{proof}
For the proof of Lemma 1, see \cite{BLF}.
\end{proof}

\end{lemma}
To ensure constraint satisfaction on the trajectory tracking error, consider a BLF

\begin{align}
    &V_1(e) \triangleq \frac{1}{2} \log{ \frac{k_b^{'2}}{k_b^{'2}-e^TPe}}
\end{align}
defined on the set $\Omega^{'}_e:\{e\in\mathbb{R}^n: e^TPe \leq k_b^{'^2}\}$, where $k^{'}_b=k_b\sqrt{\lambda_{min}\{P\}}$. If $e^TPe\rightarrow k^{'^2}_b$, i.e.
when the constrained state $e(t)$ approaches the boundary of the safe set, the BLF $V_1(e)\rightarrow \infty$, guaranteeing the safety of the system.
The unconstrained states involve continuously differentiable and positive definite quadratic functions. 



Consider the candidate Lyapunov function $V(\mu):\xi\times \mathbb{R}\rightarrow \mathbb{R}$ as,
\begin{align}
    V(\mu)&=
    \frac{1}{2}
    \bigg[\log\frac{k^{'^2}_b}{k^{'^2}_b-e^TPe}+e_d^TPe_d+tr(\tilde{K}_x^T\Gamma_x^{-1}\tilde{K}_x)\nonumber\\
    &+tr(\tilde{K}_r^T\Gamma_r^{-1}\tilde{K}_r)+tr(K_d^T\Gamma_d^{-1}K_d) +tr(K_1^T\Gamma_1^{-1}K_1)\bigg]
    \label{lyap}
\end{align}
where $\Gamma_d \in \mathbb{R}^{n \times n}$ and $\Gamma_1 \in \mathbb{R}^{n \times n}$ are positive-definite matrices. Taking the time-derivative of $V$ along the system trajectory 
\begin{align}
    \dot{V}=&\frac{1}{2(k^{'^2}_b-e^TPe)}\bigg[e^TP(A_re+B\tilde{K}_xx+B\tilde{K}_rr+B\Delta u) \nonumber\\
    &+(A_re+B\tilde{K}_xx+B\tilde{K}_rr+B\Delta u)^TPe\bigg] \nonumber\\
    &+\frac{1}{2}\bigg[e_d^TP(A_re_d+B\tilde{K}_xx+B\tilde{K}_rr+k_d\Delta u) \nonumber\\
    & +(A_re_d+B\tilde{K}_xx+B\tilde{K}_rr+K_d\Delta u)^TPe_d\bigg]\nonumber\\
    &+tr(\tilde{K}_x^T\Gamma_x^{-1}\dot{\hat{K}}_x) +tr(\tilde{K}_r^T\Gamma_r^{-1}\dot{\hat{K}}_r) \nonumber\\
    &+tr(K_d^T\Gamma_d^{-1}\dot{K}_d)+tr(K_1^T\Gamma_1^{-1}\dot{K}_1) 
    \label{vdot}
    \end{align}
    Substituting $B=K_d+K_1$ in (\ref{vdot}),
    \begin{align}
    \dot{V}=& \frac{1}{2(k^{'^2}_b-e^TPe)} \bigg [e^T(A_r^TP+PA_r)e + e^TBP\tilde{K}_xx \nonumber\\
    &+ e^TBP\tilde{K}_rr +e^TP(K_d+K_1)^T\Delta u +x^T\tilde{K}_x^TB^TPe \nonumber\\
    &+ r^T\tilde{K}_r^TB^TPe + \Delta u^T(K_d+K_1)^TPe\bigg] + \frac{1}{2}\bigg[e_d^T(A_r^TP\nonumber\\
    &+PA_r)e_d +e_d^TPB\tilde{K}_xx+ e_d^TPB\tilde{K}_rr+ e_d^TPK_d\Delta u \nonumber\\
    &+x^T\tilde{K}_x^TB^TPe_d+r^T\tilde{K}_r^TB^TPe_d +\Delta u^T K_d^Te_d\bigg]\nonumber\\
    &+tr(\tilde{K}_x^T\Gamma_x^{-1}\dot{\hat{K}}_x) +tr(\tilde{K}_r^T\Gamma_r^{-1}\dot{\hat{K}}_r) +tr(K_d^T\Gamma_d^{-1}\dot{K}_d) \nonumber\\
    &+tr(K_1^T\Gamma_1^{-1}\dot{K}_1)
\end{align}
Adaptive update laws are defines as
\begin{align}
    &\dot{\hat{K}}_x=-\bigg[\frac{\Gamma_xB^TPex^T}{k^{'^2}_b-e^TPe}+ \Gamma_xB^TPe_dx^T\bigg]\nonumber\\
    &\dot{\hat{K}}_r=-\bigg[\frac{\Gamma_r B^TPer^T}{k^{'^2}_b-e^TPe}+\Gamma_rB^TPe_dr^T \bigg ]\nonumber\\
    &\dot{K}_d=-\bigg[\frac{\Gamma_dPe\Delta u^T}{k^{'^2}_b-e^TPe}+\Gamma_dPe_d\Delta u^T \bigg]\nonumber\\
    &\dot{K}_1=-\frac{\Gamma_1 Pe\Delta u^T}{k^{'^2}_b-e^TPe}
    \label{proposedlaw}
\end{align}
which yields
\begin{align}
    \dot{V}=-\frac{1}{2}\bigg(\frac{e^TQe}{k^{'^2}_b-e^TPe}+e_d^TQe_d\bigg) \leq 0
    \label{lyapfunc}
\end{align}
which is a negative semi-definite function. 


\begin{theorem}
Consider the linear time-invariant plant (\ref{plant}) and reference model (\ref{ref}). Given Assumptions 1-3, the proposed controller (\ref{pc1}), (\ref{pc2}) and the adaptive laws (\ref{proposedlaw}) ensure that the following properties are satisfied.
\begin{enumerate}
    \item[(i)]   The plant states remain within the user-defined safe set given by $\Omega_x:=\{x\in \mathbb{R}^n:\|x(t)\|\leq \beta\}.$ 
    \item[(ii)] The control effort is bounded within a user-defined safe set given by $\Omega_u := \{u\in\mathbb{R}^{m}: \|u\|\leq u_{max}\}$.
    \item[(iii)]  All the closed loop signals remain bounded. 
    \item[(iv)] The trajectory tracking error converges to zero asymptotically i.e. $e(t)\rightarrow 0$ as $t \rightarrow \infty$.
\end{enumerate}
\end{theorem}
\begin{proof}
(i) $V(\mu)$ in (\ref{lyap}) is positive definite and $\dot{V}(\mu)\leq 0$ from (\ref{lyapfunc}), which implies that $V(\mu(t))\leq V(\mu(0))$ $\forall t\geq0$. As $V(\mu)$ is defined in the region $\Omega^{'}_e:= \{[e^T , \xi^T] \in \Psi: e^TPe\leq k^{'^2}_b \}$, it can be be inferred from Lemma 1 that 
\begin{align}
  &e^TPe\leq k^{'^2}_b  \\
  \implies &e^TPe\leq \lambda_{min}\{P\}k_b^2
  \label{lambda1}
\end{align}
Now, for any positive-definite matrix $P$,
 \begin{align}
     e^TPe\geq \lambda_{min}\{P\}\|e\|^2
     \label{lambda2}
 \end{align}
Given Assumption 3, from (\ref{lambda1}) and (\ref{lambda2}) it can be proved that 
\begin{align}
    \|e\|\leq k_b && \forall t \geq 0
\end{align} 
 i.e. the trajectory tracking error will be constrained within the user-defined safe set : $e(t) \in \Omega_e$ $\forall t\geq0$.\\
Further, since $x(t)=e(t)+x_r(t)$ and the reference model states and the trajectory tracking error is bounded, i.e. $\|x_r(t)\|\leq \alpha_1$, $\|e(t)\|\leq k_b$, it can be easily shown that the proposed controller guarantees the plants states to be bounded within the user defined safe set
\begin{align}
    \|x(t)\|\leq k_b+\alpha_1= \beta && \forall t \geq 0
\end{align}
Thus the state constraint gets satisfied, i.e. $x(t)\in \Omega_x$ for all $t\geq 0$.\\

(ii) The control effort of the proposed controller $u(t)=[u_1(t),\hdots, u_m(t)]^T$ and $\|u(t)\|=\sqrt{u_1^2(t)+u_2^2(t)+\hdots+u_m^2(t)}$. For constraining the control input two cases have been considered.\\
\textit{Case 1:} $\|v_i(t)\|\leq \frac{u_{max}}{\sqrt{m}}$\\
For this case, $u_i(t)=v_i(t)$ and $\Delta u(t)=0$.
So, $|u_i|\leq \frac{u_{max}}{\sqrt{m}}$ which implies $\|u\|<u_{max}$\\
\textit{Case 2:} $\|v_i(t)\|> \frac{u_{max}}{\sqrt{m}}$\\
For this case, $u_i(t)=\frac{u_{max}}{\sqrt{m}}sgn(v_i(t))$ which proves $\|u\|<u_{max}$.\\

(iii) Since the closed loop tracking error as well as the controller parameter estimation errors remain bounded and $K_x(t)$ and $K_r(t)$ are constants, it can be concluded that the estimated parameters are also bounded i.e. $\hat{K}_x(t), \hat{K}_r(t) \in \mathcal{L}_{\infty}$ followed by ensuring the plant state $x(t)$ and control input $u(t)$ to be bounded for all time instances. Thus the proposed controller guarantees all the the closed loop signals to be bounded.\\


(iv) Since $V(\mu)>0$ and $\dot{V}(\mu)$ is negative semi-definite (\ref{lyapfunc}), it can be shown that 
$e$, $\tilde{K}_x$, $\tilde{K}_r$, $K_d$, $K_1$ $\in \mathcal{L}_{\infty}$, $x(t)\in \mathcal{L}_{\infty}$, and $\hat{K}_x, \hat{K}_r \in \mathcal{L}_{\infty}$. Further, from (\ref{lyapfunc}) it can be shown that $e(t)\in \mathcal{L}_2$ and from (\ref{edot}) it can be inferred that $\dot{e}(t)\in \mathcal{L}_{\infty}$. Therefore, $e(t)$ is uniformly continuous. Consequently, using Barbalat's Lemma \cite{slotine}, it can be proved that $e(t)$ converges to zero asymptotically as $t \rightarrow \infty$.
\end{proof}

\section{Simulation Results}
To demonstrate the efficacy of the proposed algorithm, a multivariable LTI plant and reference model are considered.

\begin{equation*}
\scriptstyle 
\begingroup 
\setlength\arraycolsep{-1pt}
A=\begin{pmatrix}
    \scriptstyle  -0.322 & \scriptstyle 0.064& \scriptstyle  0.0364 & \scriptstyle -0.9917 &\scriptstyle  0.0003 & \scriptstyle  0.0008 &\scriptstyle  0\\
\scriptstyle 0 & \scriptstyle 0& \scriptstyle  1&\scriptstyle  0.0037 &\scriptstyle 0 &\scriptstyle 0 &\scriptstyle 0\\
   \scriptstyle -30.6492 &\scriptstyle 0 &\scriptstyle -3.6784 &\scriptstyle 0.6646 &\scriptstyle  -0.7333 & \scriptstyle 0.1315 &\scriptstyle  0\\
   \scriptstyle 8.5396 &\scriptstyle  0 &\scriptstyle -0.0254 &\scriptstyle  -0.4764 
   &\scriptstyle -0.0319 &\scriptstyle -0.0620 &\scriptstyle  0\\
   \scriptstyle 0&\scriptstyle  30&\scriptstyle  0 & \scriptstyle 10 & \scriptstyle 20.2 &\scriptstyle  0& \scriptstyle  0\\
   \scriptstyle -1 &\scriptstyle  0 & \scriptstyle  10 & \scriptstyle  0 & \scriptstyle 0 & \scriptstyle 10.25 &\scriptstyle 0\\
   \scriptstyle 0 & \scriptstyle 0 &\scriptstyle 0 & \scriptstyle 12.2958 &\scriptstyle  0 & \scriptstyle  0 & \scriptstyle -1
    \end{pmatrix} \endgroup 
    \begingroup 
\setlength\arraycolsep{-2pt}    
    \scriptstyle B=\begin{pmatrix}
    \scriptstyle 0 & \scriptstyle 0\\
    \scriptstyle 0 & \scriptstyle 0\\
    \scriptstyle 0 & \scriptstyle 0\\
    \scriptstyle 0 & \scriptstyle 0\\ 
    \scriptstyle 10.1 & \scriptstyle 0\\
    \scriptstyle 0 & \scriptstyle -4.25\\
    \scriptstyle 0 & \scriptstyle 0
\end{pmatrix} 
\endgroup
\end{equation*}

\begin{equation*}
\begingroup 
\setlength\arraycolsep{-1pt}    
\scriptstyle A_r=\begin{pmatrix}
    \scriptstyle  -0.322 & \scriptstyle 0.064& \scriptstyle  0.0364 & \scriptstyle -0.9917 &\scriptstyle  0.0003 & \scriptstyle  0.0008 &\scriptstyle  0\\
\scriptstyle 0 & \scriptstyle 0& \scriptstyle  1&\scriptstyle  0.0037 &\scriptstyle 0 &\scriptstyle 0 &\scriptstyle 0\\
   \scriptstyle -30.6492 &\scriptstyle 0 &\scriptstyle -3.6784 &\scriptstyle 0.6646 &\scriptstyle  -0.7333 & \scriptstyle 0.1315 &\scriptstyle  0\\
   \scriptstyle 8.5396 &\scriptstyle  0 &\scriptstyle -0.0254 &\scriptstyle  -0.4764 
   &\scriptstyle -0.0319 &\scriptstyle -0.0620 &\scriptstyle  0\\
   \scriptstyle 0&\scriptstyle  0&\scriptstyle  0 & \scriptstyle 0 & \scriptstyle -20.2 &\scriptstyle  0& \scriptstyle  0\\
   \scriptstyle 0 &\scriptstyle  0& \scriptstyle  0& \scriptstyle  0 & \scriptstyle 0 & \scriptstyle -20.2 &\scriptstyle 0\\
   \scriptstyle 0 & \scriptstyle 0 &\scriptstyle 0 & \scriptstyle 12.2958 &\scriptstyle  0& \scriptstyle  0 & \scriptstyle -1
    \end{pmatrix} \endgroup
    \begingroup 
\setlength\arraycolsep{-2.5pt}    
    \scriptstyle B_r=\begin{pmatrix}
    \scriptstyle 0 & \scriptstyle 0\\
    \scriptstyle 0 & \scriptstyle 0\\
    \scriptstyle 0 & \scriptstyle 0\\
    \scriptstyle 0 & \scriptstyle 0\\ 
    \scriptstyle 20.2 & \scriptstyle 0\\
    \scriptstyle 0 & \scriptstyle 20.2\\
    \scriptstyle 0 & \scriptstyle 0
\end{pmatrix}
    \endgroup
\end{equation*}
\\~\\
The reference signal is considered as: $r(t)=[\exp(-t/10);\exp(-t/20)]$. The other parameters are chosen as: $\Gamma_x=\begin{bmatrix}
5 & 0\\
0& 5
\end{bmatrix}$, $\Gamma_r=\begin{bmatrix}
5 & 0\\
0& 5
\end{bmatrix}$, $\Gamma_1=I_{n \times n}$, $\Gamma_d=I_{n \times n}$, 
$u_{max}=2.5$, $\beta=2$, $\alpha_1=1.5$ and $k_b=0.5$.\\
The desired controller must satisfy the input constraint $\|u\|\leq2.5$, while ensuring that plant states within user-defined bound i.e. $\|x\|\leq2$. Given Assumption 1, $\|x_r\| \leq 1.5 $, the state constraint is equivalent to satisfying the constraint in the error i.e. $\|e\|\leq0.5$. It is assumed that the initial plant states remain within the user-defined safe set i.e. $\|x(0)\|\leq2$.\\
To gauge the safety and performance of the proposed control law, we compare it with the classical MRAC controller (in (\ref{ueq}) and (\ref{MRAC})). The reference signal is considered as: $r(t)=[\exp(-t/10);\exp(-t/20)]$. The adaptation gains are chosen as: $\Gamma_x=\begin{bmatrix}
25 & 0\\
0& 25
\end{bmatrix}$, $\Gamma_r=\begin{bmatrix}
25 & 0\\
0& 25
\end{bmatrix}$.\\
Note that, adaptation gains are tuned to achieve better tracking performance for both proposed controller and classical MRAC. 
\begin{figure}[h!]
\centering
\includegraphics[width=\linewidth]{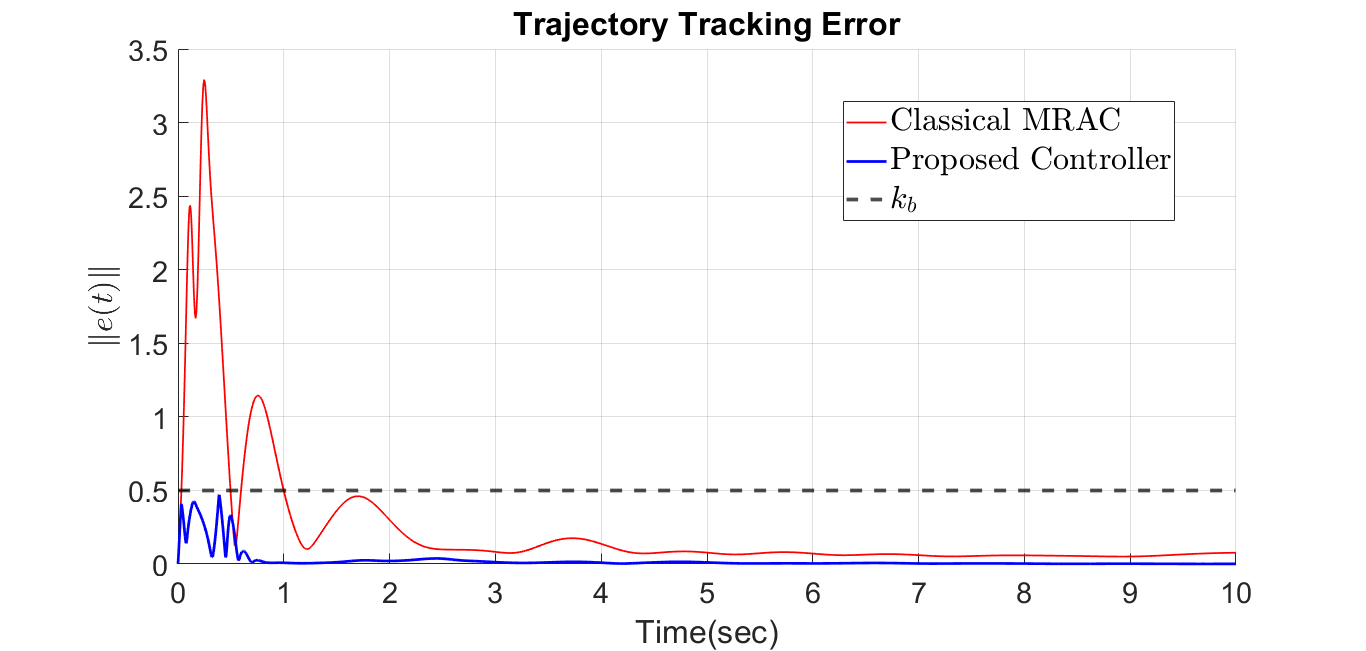}
\caption{Comparative analysis of trajectory tracking error using the proposed control law (\ref{proposedlaw}) and classical MRAC law (\ref{MRAC}).}
\end{figure}

\begin{figure}[h!]
\centering
\includegraphics[width=\linewidth]{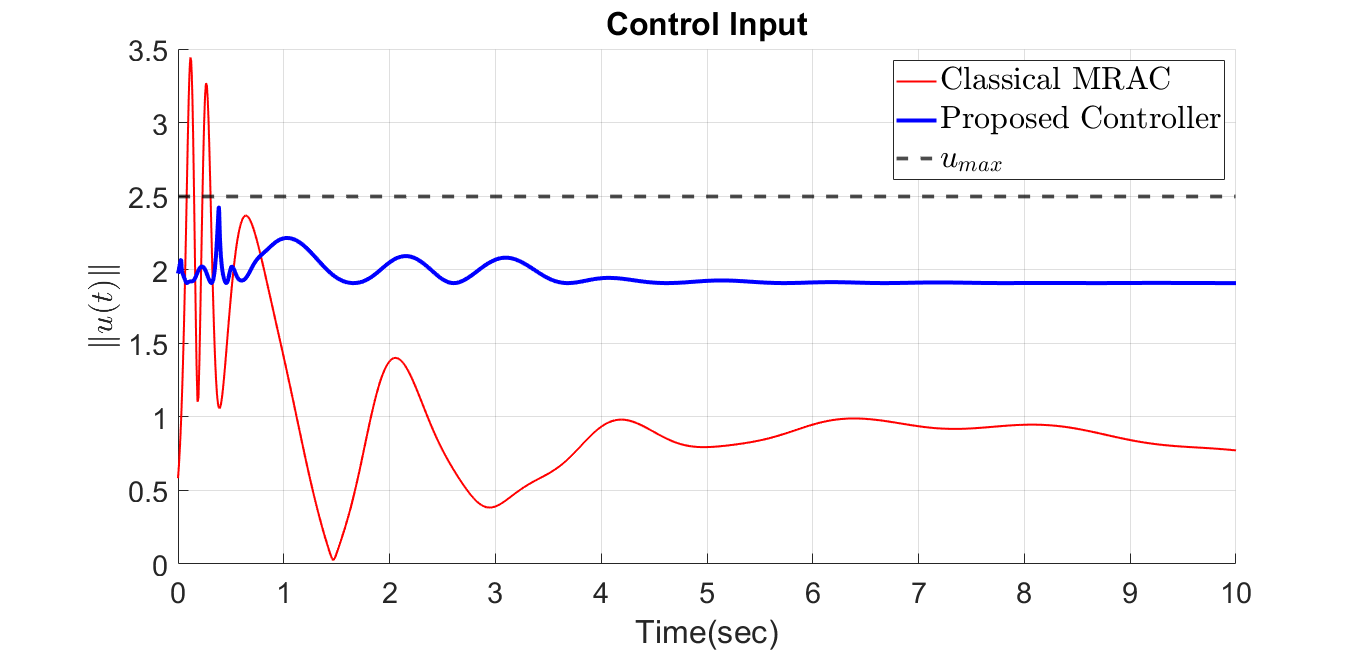}
\caption{Control input using proposed controller (\ref{proposedlaw}) and classical MRAC (\ref{MRAC}).}
\end{figure}
Fig.1 shows the trajectory tracking error using the proposed method where the user-defined constraint is satisfied while in conventional MRAC case, the norm of the trajectory tracking error goes beyond the safe region. Furthermore, the proposed controller ensures that the tracking error converges to zero as time tends to infinity and the rate of convergence is higher than the classical MRAC.
The proposed control architecture bounds the control effort in user-defined constrained region (Fig. 2)  while for the conventional MRAC the bound on the control input can not be known \textit{a-priori}. Fig. 3 shows the state trajectories of the plant and reference model.
\begin{figure}[h!]
\centering
\includegraphics[width=\linewidth]{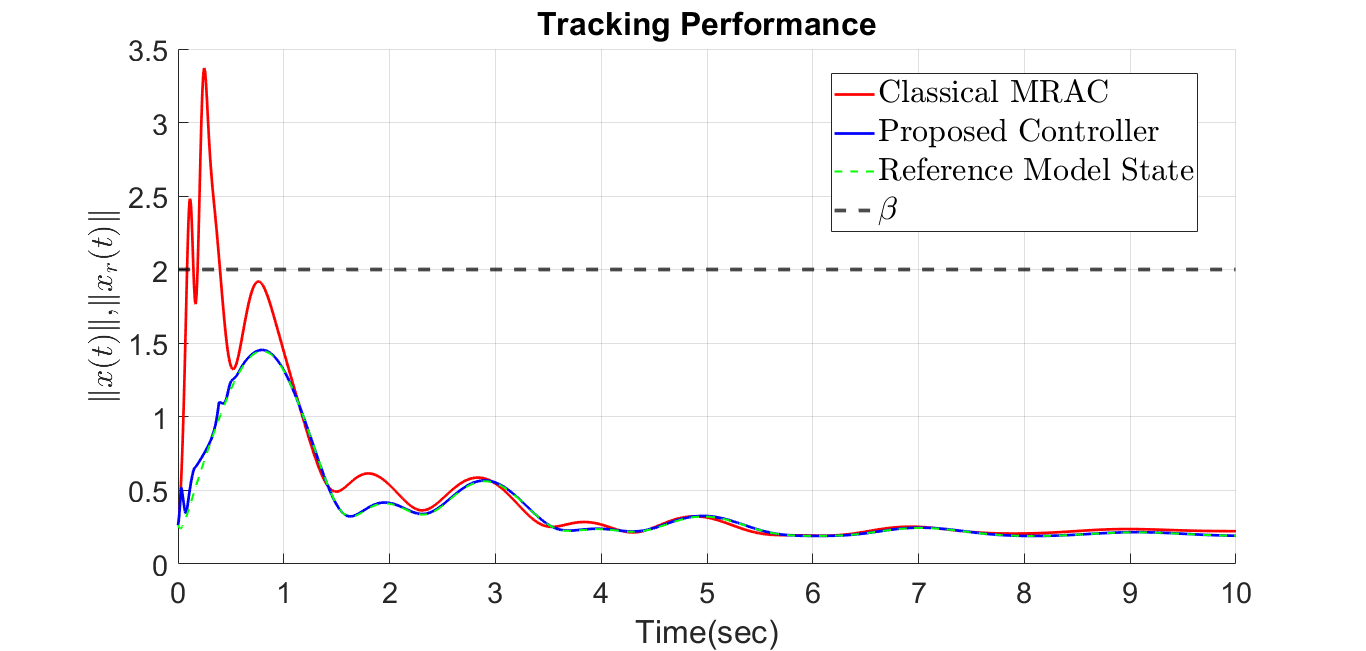}
\caption{Tracking performance of the plant using proposed controller (\ref{proposedlaw}) and classical MRAC (\ref{MRAC}).}
\end{figure}





 It is seen that increasing the adaptation gain leads to better tracking performance, in general for both the classical and the proposed controller, although the response becomes more oscillatory. The improved tracking performance of the classical MRAC is achieved at the cost of greater control effort leading to violation of the input constraints. Further, in the classical MRAC case, the high frequency oscillation in the control input may even violate the actuation rate limits. On the other hand, the state and the input constraints are never violated in case of the proposed controller.


\section{Conclusion}
In this paper, a  novel MRAC architecture is proposed for multivariable LTI systems by strategically combining BLF with a saturated controller which guarantees  both the plant state and the control input remain bounded within user-defined safe sets. The proposed controller also ensures that the trajectory tracking error asymptotically converges to zero and the closed-loop signals remain bounded. Simulation studies validate the efficacy of the proposed control law comparing to classical MRAC.
Extending the work to uncertain nonlinear systems and exploring robustness properties is an important area of future research.

\bibliographystyle{ieeetr}
\bibliography{ref}

\end{document}